\documentclass[10pt]{article}
\usepackage{dcolumn}
\usepackage{bm}
\usepackage{verbatim}       
\usepackage{hyperref}

\usepackage{amsthm}
\usepackage{amssymb}

\usepackage{bm}
\usepackage{amstext}
\usepackage{psfrag}
\usepackage{amsmath}
\usepackage{amsfonts}
\usepackage{units}
\usepackage{mathrsfs}
\newcommand{\p}{\partial}

\newcommand{\dd}{{\rm d}}

\newcommand{\bd}{\begin{definition}}                
\newcommand{\ed}{\end{definition}}                  
\newcommand{\bc}{\begin{corollary}}                 
\newcommand{\ec}{\end{corollary}}                   
\newcommand{\bl}{\begin{lemma}}                     
\newcommand{\el}{\end{lemma}}                       
\newcommand{\bp}{\begin{proposition}}            
\newcommand{\ep}{\end{proposition}}                
\newcommand{\bere}{\begin{remark}}                  
\newcommand{\ere}{\end{remark}}                     

\newcommand{\bt}{\begin{theorem}}
\newcommand{\et}{\end{theorem}}

\newcommand{\be}{\begin{equation}}
\newcommand{\ee}{\end{equation}}

\newcommand{\bit}{\begin{itemize}}
\newcommand{\eit}{\end{itemize}}
\newtheorem{theorem}{Theorem}[section]
\newtheorem{corollary}[theorem]{Corollary}
\newtheorem{lemma}[theorem]{Lemma}
\newtheorem{proposition}[theorem]{Proposition}
\theoremstyle{definition}
\newtheorem{definition}[theorem]{Definition}
\theoremstyle{remark}
\newtheorem{remark}[theorem]{Remark}
\newtheorem{example}[theorem]{Example}

\begin{document}

\title{Horizon classification via  Riemannian flows\footnote{In this web version we fixed a few more typos with respect to the version published in Annales Henri Poincar\'e. Also, we added a couple of  explanatory footnotes. The content does not change.}}

\author{R. A. Hounnonkpe\footnote{Universit\'e d'Abomey-Calavi, B\'enin and Institut de Math\'ematiques et de Sciences Physiques (IMSP), Porto-Novo, B\'enin. E-mail:
rhounnonkpe@ymail.com} \ and \
E. Minguzzi\footnote{Dipartimento di Matematica, Universit\`a degli Studi di Pisa,  Largo B. Pontecorvo 5,  I-56127 Pisa, Italy. E-mail:
ettore.minguzzi@unipi.it}}

\date{}

\maketitle

\begin{abstract}
\noindent We point out that the geometry of connected totally geodesic compact null hypersurfaces in Lorentzian manifolds is only slightly more specialized than that of Riemannian flows over compact manifolds, the latter mathematical theory having been much studied in the context of foliation theory since the work by Reinhart (Ann Math
69:119, 1959). We are then able to import  results on Riemannian flows to the horizon case,  so obtaining theorems on the dynamical structure of compact horizons that do not rely on (non-)degeneracy assumptions. Furthermore, we clarify the relation between isometric/geodesible Riemannian flows and non-degeneracy conditions.
This work also contains some positive results on the possibility of finding, in the degenerate case, lightlike  fields tangent to the horizon that have  zero surface gravity.
\end{abstract}


\section{Compact horizons and Riemannian flows}

According to the Penrose's strong cosmic censorship  (SCC) conjecture \cite{penrose69,penrose79} the maximal globally hyperbolic development of generic initial data for the Einstein equations is inextendible. In other words, it roughly states that generically, determinism holds in GR, i.e.\ that the future is determined by the past. The abundance of exact spacetimes that present Cauchy horizons would be due to their non-genericity. If SCC is true, then, to start with, {\em compact} Cauchy horizons should not form save for non-generic i.e., very special, spacetimes.

This expectation has been formalized in the  Isenberg-Moncrief conjecture \cite{moncrief83}, which roughly states that if a spacetime admits a non-empty
compact Cauchy horizon, then the spacetime presents a Killing symmetry (hence might be rightly regarded as non-generic), and the Cauchy horizon is  a Killing horizon.

Isenberg and Moncrief studied this problem  solving it in the positive already in their first work under (a) analyticity assumptions on the spacetime metric and the horizon; (b) (electro-)vacuum conditions; (c) closedness of the horizon generators.

Since their first work, they conjectured that in (a) analyticity could be weakened to smoothness and that condition (c) could be dropped. Concerning (a), it was shown that there is no need to assume smoothness of the horizon as it follows from the smoothness of the metric under the null energy condition \cite{larsson14,minguzzi14d}. Under such energy assumption, compact Cauchy horizons are smooth totally geodesic null hypersurfaces, so one often focuses the study on these objects.

The total geodesic property for the horizon $H$ reads as follows
\begin{equation} \label{pprt}
\nabla_X n = \omega(X) n,
\end{equation}
where $X$ is an arbitrary field tangent to the horizon, $n$ is the future-directed lightlike vector field tangent to the horizon and  $\omega: H \to T^*H$ is a 1-form. The connection $\nabla$ can be regarded as the Levi-Civita connection of the spacetime metric or as the induced connection on the horizon. There is no ambiguity precisely because of the total geodesic property. The {\em surface gravity} is the quantity
\[
\kappa:=\omega(n).
\]
The tangent field can be rescaled $n'=e^f n$, $f:H \to \mathbb{R}$, which redefines the 1-form and surface gravity as follows
\begin{align}
    &\omega'=\omega+\mathrm{d}f, \label{ren1}\\
    &\kappa'=e^f(\kappa+ n(f) ). \label{ren2}
\end{align}

A natural problem is that of establishing whether surface gravity can be normalized to a constant,   as this is a necessary condition for the horizon to be Killing \cite{chrusciel20}. Moreover, under the constancy of surface gravity, it is not hard to show that the Lie derivative of the metric at the horizons vanishes, $L_n g=0$. The problem of extending $n$ to a Killing vector field in a neighborhood of the horizon becomes then approachable through PDE methods with the idea of  propagating the condition $L_ng=0$.

Let us also call a horizon {\em degenerate} if it admits a complete generator, and {\em non-degenerate} if it admits an incomplete generator. Isenberg and Moncrief also proved the {\em dichotomy}, i.e.\ that under the above assumptions ((a), (b) and (c)), the horizon is either degenerate, in which case all generators are complete, or non-degenerate in which case all generators are future incomplete (or dually).

In a long series of works \cite{moncrief83,isenberg85,isenberg92,moncrief08,moncrief20} spanning almost forty years they weakened (a) or (c) and used some alternative combinations and variants of these properties, or assumptions of non-degeneracy, to get stronger versions. The original problem remained open but in the process they developed some novel and useful techniques among which are the {\em Gaussian null coordinates} and the {\em  ribbon argument}.

More recently, in a work by Bustamante and Reiris \cite{reiris21} and in a subsequent work by Gurriaran and the second author \cite{minguzzi21},  by using the ribbon argument the dichotomy was proved in the smooth category and, also, it was proved that in  the non-degenerate case surface gravity can be normalized to a non-zero constant  (if surface gravity is normalizable to zero then the horizon is degenerate). In \cite{minguzzi21}  is was shown that condition (b) can be dropped provided the dominant energy condition is assumed.

The extension PDE step was obtained, still in the non-degenerate case, in works by Petersen \cite{petersen18b,petersen19} and Petersen and Racz \cite{petersen18}, so that the Isenberg-Moncrief conjecture is now solved for non-degenerate horizons. In this case, Bustamante and Reiris have also obtained a detailed classification result for the dynamics of the generators in 4 spacetime dimensions by using the theory of Lie group actions over compact manifolds \cite{bustamante21} (see also \cite{kroenke24} for other geometrical results for non-degenerate horizons but in the analytic setting).

It is apparent from the previous summary of available results that the Isenberg-Moncrief conjecture remains widely open in the degenerate case.\footnote{Studies in the near horizon geometry of extremal stationary black holes \cite{kunduri13} assume that the horizon is Killing from the outset so they are not really helpful in proving the conjecture though they are important for clarifying the classification of horizons.}
With this work we aim to change this situation by obtaining some useful and significant results.


We end this section by introducing some definitions and terminology which are the same of \cite{minguzzi14d,minguzzi18b}. A spacetime
$(M, g)$ is a paracompact, time oriented Lorentzian manifold of dimension $n + 1 \ge 2$.
The signature of the metric is $(-,+,\ldots,+)$.
We  assume that  $M$ is  $C^k$, $4 \le k \le \infty$, or even analytic, and so as it is $C^1$ it has a unique $C^\infty$ compatible structure (Whitney) \cite[Theor.\ 2.9]{hirsch76}.
Thus we could assume that $M$ is smooth without loss of generality.

The metric will be assumed to be $C^3$, but it is likely that the degree can be lowered. We assume at least this degree of differentiability because it was also assumed in  \cite{minguzzi14d} over which results we rely.
For shortness, sometimes we shall sloppily use the word {\em smooth} as meaning {\em as much as the regularity of the metric allows}, when such a regularity is clear.

\section{The Riemannian flow structure of horizons}

For shortness, we shall use the following terminology\\

\begin{definition}
 A  {\em horizon} is  a connected totally geodesic $C^3$ null hypersurface. \\
\end{definition}

With $H$ we shall denote a horizon in the above sense. With $n$ we denote a future-directed lightlike vector field tangent to the horizon.

Concerning the totally geodesic property, we recall that it is equivalent to the vanishing of the second fundamental form, which is defined as an endomorphism $X\mapsto b(X):= \overline{\nabla_Xn}$ on the quotient bundle $TH/n$, where the overline denotes projection under that quotient. It can be shown that it vanishes iff $\nabla_XY\in TH$ for any two vector fields tangent to the horizon. This equivalence goes back to Kupeli \cite[Thm.\ 30]{kupeli87}. The condition $b=0$ can also be written as Eq.\ (\ref{pprt}) while the property ``$\nabla_XY\in TH$ for any two vector fields tangent to the horizon'' makes it possible to regard $\nabla$ both as the spacetime connection and as the connection induced on $H$.

We shall be mostly interested in the compact horizon case but many of the results that follow do not depend on this assumption, so this property will be spelled out  when needed.


We denote the metric induced on the horizon with $g_T$ and call it {\em transverse metric}.  Forgetting about the spacetime, a horizon can be regarded as a triple $(H,g_T,\nabla)$  where $g_T$ is, at every point,  a positive semi-definite symmetric bilinear  form with 1-dimensional oriented kernel $\textrm{span}(n)$, $n$ is a $C^2$ (positive) vector field on $H$ such that $g_T(n,
\cdot)=0$, $\nabla$ is an affine connection, and we have the compatibility conditions $\nabla g_T =0$, $\nabla n=n \otimes
\omega$ for some 1-form $\omega$. By using these properties it follows that the  flow of  $n$  preserves $g_T$ (e.g.\ \cite[Thm. 37]{kupeli87} \cite[Lemma B1]{friedrich99}\cite{moncrief08}\cite[Lemma 7]{minguzzi21}), that is for $X,Y\in TH$
\[
L_n g_T(X,Y)\!=g_T(\nabla_X n, Y)+g_T(X, \nabla_Y n)\!=\omega(X) g_T(n,Y)+\omega(Y) g_T(X,n)=0.
\]
Observe that the kernel $\textrm{span}(n)$ being 1-dimensional is necessarily integrable. Its leaves are the (unparametrized) horizon generators.

Now, let us pause and introduce a different concept.

A foliation $(H,\mathcal{F})$ in a manifold is a completely integrable distribution $D$ of constant dimension. An integral leaf is denoted $L$ and the set of leaves is $\mathcal{F}$. Two foliations $(M,\mathcal{F})$, $(M',\mathcal{F}')$ are {\em conjugate} if there is a diffeomorphism $M\to M'$ that sends leaves to leaves.

Suppose that on a manifold $H$ there is a positive semi-definite symmetric bilinear form  $g_T:H\to T^*H\otimes_H T^*H$  with a kernel
\[
\ker g_T:=\{X\in TH: g_T(X,\cdot)=0\}
\]
 of constant dimension,  and that $L_Xg_T=0$ for every $X \in \Gamma(\ker g_T)$. Then the distribution $\ker g_T$ is integrable and defines a foliation $\mathcal{F}$ on $H$ \cite[Prop.\
 3.2]{molino88}.

The foliation $(H,\mathcal{F})$ is said to be {\em Riemannian} if there is a metric $g_T$ with the properties of the previous paragraph from which the foliation can be recovered  \cite[Sec.\
 3.2]{molino88}.
Sometimes by {\em Riemannian foliation} one understands the pair  $(H, g_T)$  if a metric $g_T$ which satisfies the above properties is assigned.

 A {\em Riemannian flow} is a Riemannian oriented 1-dimensional foliation  \cite{carriere84}. We denote with $n$ a non-vanishing positively oriented field tangent to the foliation. It is defined up to positive rescalings. The Riemannian flow condition requires the existence of $g_T$ such that $L_n g_T=0$ for some choice of $n$. Indeed,
 \begin{align*}
 (L_{fn} g_T)(X,Y)&={ fn(g_T(X,Y))}-g_T(L_{fn}X, Y)-g_T(X,L_{fn} Y)\\
 &={f}(L_n g_T)(X,Y)+X(f)g_T(n,Y)+Y(f) g_T(X, n)=0.
\end{align*}
 Not all oriented 1-dimensional foliations are Riemannian flows \cite{molino88}.

The first statement in the following result is clear\\

\begin{proposition} \label{bort}
Every horizon is, in particular, a Riemannian flow with a preassigned $g_T$.
Every smooth null hypersurface $H$ that induces  a Riemannian flow structure $(H,g_T)$ is necessarily totally geodesic.
\end{proposition}

\begin{proof}
The last statement follows from the Kupeli's result \cite[Thm. 37]{kupeli87}: for a null hypersurface $H$ the condition $L_n g_T=0$  is equivalent to the totally geodesic property. Indeed, for $X,Y\in \mathfrak{X}(H)$, denoting with $\nabla$ the spacetime connection, and using the fact that it is torsionless
\begin{align*}
L_n g_T(X,Y)&=g_T(\nabla_X n, Y)+g_T(X, \nabla_Y n)=-[g(n, \nabla_X Y+\nabla_Y X)]\\
&=-2g(n, \nabla_X Y) ,
\end{align*}
which implies $L_n g_T=0$ iff $\nabla_XY \in TH$.
\end{proof}
In horizon geometry it is possible to define the surface gravity as $\kappa:=\omega(n)=\textrm{tr} \nabla n$, while it does not seem to be directly definable in the geometry of Riemannian flows. Still, part of our problem will be to relate results on surface gravity to the  Riemannian flow structure of the horizon.

The theory of Riemannian foliations began with the work by Reinhart \cite{reinhart59} (1959) and has been presented in several books, including those by Molino \cite{molino88} and Tondeur \cite{tondeur97}.
The purpose of this work is to apply the theory of  Riemannian flows to  horizons. In fact, the results on Riemannian foliations/flows are so well adapted to our problems and yet so well established that it is surprising that they have not been noticed and applied before to horizon geometry. We shall show that several results discovered over the years in horizon geometry were essentially already known by the community working in  Riemannian foliations. We believe that making manifest these connections might be as valuable as working out entirely new original results, as  communication between two different fields might result in a faster development for both of them.

We start presenting some key results for Riemannian flows. Some of them  hold  for Riemannian foliations, but in our context it does not seem particularly useful to retain complete generality. In fact, we phrase them in terms of the Riemannian flow structure of the horizon for direct application in general relativity.

Let us consider the horizon as a foliated manifold, where the leaves are the generators, hence generated by a distribution $D:=\textrm{span}(n)$.
By continuity of $n$, $D$ is a closed subbundle of $TH$. A foliation is {\em simple} if the leaves are the connected fibers of a submersion $H\to W$ where $W$ is a manifold (manifolds are Hausdorff in our terminology). Every foliation is locally simple. For instance,  a flow is covered by charts that are cylindrical subsets  $U:=I \times B \subset \mathbb{R}\times \mathbb{R}^q$, $q=n-1$, $B$ representing the local quotient and the affine real line intervals giving the local leaves (also called {\em plaques} in this context). In general, the space of leaves is not a manifold.

A set $S\subset H$ is {\em saturated} if it is the union of leaves. The {\em saturation} of a subset $A$ is the union of all the leaves it intersects, and it is  saturated. A perhaps non-obvious fact from foliation theory \cite[Prop.\ 1.3]{molino88} states that the closure and interior of a saturated set are saturated, and that the saturation of an open set is open.

An immersed submanifold $T$ of $H$ of dimension $q=n-1$ is a {\em transversal} (this is a local concept) if its tangent space is supplementary to  $D$ at every point. It is a {\em total transversal} if it intersects every leaves.\footnote{One should not confuse this notion of immersed transverse submanifold, which might not exist, with the so-called {\em transverse manifold} of a foliation. The latter always exists and can be constructed from the disjoint sum of local transversal, identified with open sets of $\mathbb{R}^q$. This transversal comes with differentiable maps between the mentioned sets that constitute it, which are actually isometries in the case of a Riemannian foliation. This point of view is due to Haefliger \cite[Chap.\ 1]{tondeur97}.} A compact total transversal $T$ is a  {\em cross-section} if for every $p \in H$ there is $t >0$ such that $\varphi^n_t(p)\in T$ where $\varphi^n$ is the flow of $n$. If $H$ is compact a compact total transversal is a cross-section \cite{fried82}.

 A foliation rarely admits a total transversal, but this happens for foliations obtained by an important procedure known as {\em suspension} \cite{smale67}\cite[p.\ 28]{molino88}. It starts from two manifolds $A$ and $T$, where $A$ has the dimension of the desired leaves, and rather than producing $H$ via a product $A\times T$, it uses the fundamental group of $\pi_1(A,x_0)$ with respect to some point $x_0\in A$ and a homomorphism $h: \pi_1(A,x_0) \to Diff(T)$ to define how to glue everything into $H$. In the flow case, $A$ has non-trivial homotopy only if it is a circle, thus we just need to  assign a diffeomorphism $\phi:T\to T$, consider the product $[0,1]\times T$ and identify for every $p\in T$, $(1,p)$ with $(0,\phi^{-1}(p))$, in other words consider the manifold $\mathbb{R}\times T/\sim$, $(t,p)\sim (t+1,\phi(p))$. Actually, as we are constructing a Riemannian flow, $T$ must be regarded as a Riemannian manifold $(T,g_T)$ and so $\phi$ has to be an isometry.  A flow on a compact manifold $H$ is a suspension iff it admits a cross-section \cite{fried82}.

 We note that this type of suspension construction is  encountered in the study of black hole horizons after a certain compactification trick is applied, particularly higher dimensional ones where the generators do not necessarily close \cite{friedrich99,hollands07}.

Since every leaf is the  continuous image of an integral curve of $n$, it is possible to attach to it a topology induced by the real line called {\em leaf topology}. If it coincides with the topology induced from $H$ we speak of {\em proper}  leaf. This happens if and only if the leaf admits a transversal at some point which intersects the leaf at just that point, in which case it admits one at every point \cite[Prop.\ 1.5]{molino88}. Every closed leaf is proper \cite{molino88} (by {\em closed leaf} we mean one  which is the image of a generator $\gamma: \mathbb{R} \to H$, with $\gamma(a)=\gamma(b)$ for some $a\ne b$, not to be confused with the {\em leaf closure} which is $\overline{\textrm{Im} \gamma}$ where the closure is in the topology of $H$).
A stronger concept is that of a {\em regular} leaf, that is a leaf that admits a transversal that  intersects each leaf at most once \cite{reinhart59}.

A vector field $X \in \mathfrak{X}(H)$ is said to be  {\em foliate} if $L_X n \propto n$, in other words, if its flow sends leaves to leaves \cite[Prop.\ 2.2]{molino88}. This property is well-posed as it is invariant under rescalings of $n$. If $X$ is foliate then $X+fn$ is foliate where $f$ is any $C^1$ function.

The quotient space  $Q=TH/D $ is a $q$-dimensional (transverse) vector bundle with base $H$. We denote with an overbar the projection $TH\to Q$. A local section $\bar X$ of $Q$ is said to {\em slide along the leaf} if it is the projection of a foliate field $X$.

It is also natural to study the bundle $B(H)$ of orthonormal reference frames of $Q$. The main idea by Molino, which allowed him to obtain structure theorems that were subsequently refined by other authors and that we shall recall in what follows, was to consider this bundle and a suitably lifted foliation over it.

A $k$-form $\omega$ over $H$ is said to be {\em basic} if it satisfies
\[
\iota_n \omega=0, \qquad \iota_n \dd \omega=0,
\]
 which is  well posed as invariant under rescalings of $n$. The second condition can also be replaced by $L_n \omega$.
The exterior differential of a basic form is basic. Reinhart \cite{reinhart59b} introduced the {\em basic cohomology} $H^k_b(\mathcal{F})$ of the (class of) closed basic $k$-forms that are not the external differential of a basic  $k-1$-form.

Another concept introduced by Reinhart is the following \cite{reinhart59}\\

\begin{definition}
A Riemannian metric $g_B$ on the foliated manifold $H$ is {\em bundle-like} if for every foliate fields $X,Y$ that are $g_B$-orthogonal to $n$  we have \[
n (g_B(X,Y))=0.
\]
\end{definition}

The following result is a simple consequence of the definitions\\

\begin{proposition}
Every bundle-like metric can be uniquely written in the form
\[
g_B=g_T+\sigma \otimes \sigma
\]
for some 1-form $\sigma$ such that $\sigma(n)>0$ and where $g_T$ is a transverse metric in the sense that it has signature $(0,+,\ldots,+)$ and $L_n g_T=0$, where $n \in \textrm{ker} g_T$. Conversely, given the transverse metric $g_T$, the expression $g_B$ above provides a bundle-like metric for any   1-form $\sigma$ such that $\sigma (n) > 0$ (these forms exist, this is clear locally, the global result being obtained with a partition of unity).\\
\end{proposition}

Giving $g_B$, the uniquely determined  1-form $\sigma$ is also called the {\em characteristic form} of $g_B$. If $n$ is $g_B$-normalized, it can be written $\sigma(\cdot):=g_B(n,\cdot)$. In other words, $\sigma$ is the 1-form that vanishes in the space orthogonal to the leafs and that measures the $g_B$-length over the oriented leaves.

The curves that are contained in a leaf are said to be {\em vertical}, while those that are at every point  $g_B$-perpendicular to the leaves, and hence have tangent in $\ker \sigma$, are called {\em horizontal}.

A bundle-like metric provides  more information than a Riemannian flow. It might look as a somewhat artificial concept for the study of the Riemannian flow structure, were not for a surprising foundational result by Reinhart on Riemannian foliation theory \cite{reinhart59} \cite[Prop.\ 3.5]{molino88} (here formulated in the flow case)\\

\begin{theorem} \label{kkr}
Consider a foliated manifold $H$ with  one-dimensional leaves  and suppose that $H$ is endowed with a bundle-like metric $g_B$. The $g_B$-geodesics  that start orthogonal to one leaf are orthogonal to all leaves they intersect. In particular, sufficiently close leaves remain at constant $g_B$-distance (this distance depends only on $g_T$).\\
\end{theorem}
Indeed, if we look at the local transverse geodesic from a simple cylindrical neighborhood, as introduced above, the geodesic is the horizontal lift of a geodesic on the local quotient $(B, g_T)$ (well, $g_T$ projects to a well defined Riemannian metric on $B$, here denoted in the same way). The  property of the theorem actually characterizes bundle-like metrics  \cite{reinhart59}. The reader acquainted with the literature on horizons will recognize that Reinhart's horizontal exponential map appeared in the work by Bustamante and Reiris \cite{reiris21}.\\

\begin{remark}
The notion of bundle-like metric was rediscovered in horizon geometry through the notion of {\em rigging} \cite{katsuno80,gutierrez16,gutierrez21}, which is a vector field $N \in TM\vert_H$ transverse to $H$. Then $\sigma:=g(N,\cdot)$ and the Riemannian metric is defined as above. Although not uniquely determined, the notion of rigging, and the derived Riemannian metric approach, proved useful in the study of horizons.  More precisely, this approach has been used for general null hypersurfaces for which $(M,g_T)$ is not necessarily a Riemannian flow structure. We recall that in our work the null hypersufaces are totally geodesic, which guarantees that they induce a Riemannian flow.\\
\end{remark}

Let us consider a closed leaf  and let $p$ be one of its points. After one cycle the vector space $Q_p$ is linearly mapped to itself by the slide along the leaves. This is the holonomy of the leaf of the foliation.\\

\begin{remark}
For non Riemannian foliations the holonomy has a  more complicated definition. One needs to introduce a transversal $T$ to the leaf starting from $p$, consider a oriented closed path on the leaf, and the germ of diffeomorphisms from $T$ to $T$ obtained via the sliding along  the leaf \cite{molino88}. In the Riemannian foliation case such germ is determined by the tangential information. Indeed, it is possible to infer the map from leaves to leaves in a local neighborhood of $p$ by introducing a bundle-like metric. Then each leaf is locally determined by the horizontal exponential map of a certain vector $X \in T_pH$,  $\bar X \in Q_p$. The horizontal exponential map image of the sliding vector $X(s)$ belongs to the same   leaf for every $s$ \cite{reinhart59},  so, in order to describe the dynamics of nearby leaves, it is really only important to keep track of how the pointing vector $X(s)$ changes in the sliding along the leaf.\\
\end{remark}

If the holonomy is trivial, i.e.\ the identity, the leaf is said to be  {\em flat}. The foliation is {\em flat} if every leaf is flat \cite{reinhart59}.

If a foliation is simple it is flat. More generally, since the slide along the leaf preserves $g_T$, ultimately the holonomy is an endomorphism of $Q_p$ that preserves $g_T(p)$, thus it is an orthogonal transformation. This transformation  depends only on the foliated homotopy of the path followed.

The main use by Reinhart of Theorem \ref{kkr} was to prove that in Riemannian foliations every two leaves have  the same covering, a fact which is trivial for flows, as the real line covers every leaf. However, he obtains results that are non-trivial also for flows. For instance, he proves that every non-closed leaf has only non-closed nearby leafs \cite[Cor.\ 2]{reinhart59} and that each component of the set of regular leaves is a fibre space over a Hausdorff manifold \cite[Cor.\ 3]{reinhart59}. Furthermore, he proves that every regular leaf is flat, and that a flat foliation that is regular at one point is regular everywhere \cite[Lemmas 5,6]{reinhart59}.

A Satake manifold $S$ is a generalization of differentiable manifold in which the local model manifolds are not the open subsets of $\mathbb{R}^q$ but rather quotient of open sets of $\mathbb{R}^q$ by some finite group of diffeomorphism, see \cite[Sec.\ 3.4]{molino88} for an introduction. Therefore, the Satake manifold is  covered by open sets that are in bijection with such model spaces. A point $x\in S$ is said to be {\em regular} if it admits in the maximal atlas an open neighborhood in bijection with a standard open set of $\mathbb{R}^q$, otherwise it is called {\em singular}. If all points were regular the Satake manifold would be an ordinary manifold. It turns out that the subset of regular points is open and dense in $S$. Furthermore, at a singular point $p$ it is possible to find in the maximal atlas an open set in bijection with the quotient of an open set of $\mathbb{R}^q$ under the action of a finite subgroup of orthogonal transformations,  the point $p$ being mapped to the origin \cite[Lemma 3.6]{molino88}.

Another important result by  Reinhart is \cite{reinhart59} \cite[Prop.\ 3.8]{molino88}\\

\begin{theorem}
Let $(H,g_T)$ be a compact Riemannian flow structure in an $n$-dimensional manifold with closed leaves (that is, foliated by circles). The leaf space $S$ admits a natural structure of $n-1$-dimensional Satake manifold and the projection $H\to S$ is a morphism of Satake manifolds. Furthermore, the set of regular points in $S$ coincides with those leaves that have trivial holonomy. If the singular leaves are removed one is left with a simple foliation, i.e.\ a circle fibration over a standard manifold which is the regular subset of the original Satake base manifold.\\
\end{theorem}

Another result that can be metioned is the fact that if, in a compact connected Riemannian flow structure,  a leaf is  closed (a circle) with trivial holonomy, then all the other leaves are circles \cite[p.\ 100 (point 4)]{molino88} and so the previous theorem applies.

It can be here recalled that by Epstein's theorem \cite{epstein72}, a compact oriented 3-manifold foliated by circles is diffeomorphic to a Seifert fibration. In the latter structure we find the same notions of regular and singular leaves which correspond to those of the above construction via the Satake quotient manifold. It should be noted that Reinhart's theorem is better adapted to our purposes as it does not impose any constraint on the dimension. Epstein's result requires a condition on dimension (luckily useful in the physical 4-dimensional spacetime case, as horizons are 3-dimensional) as it does not use the existence of the transverse metric $g_T$.

Epstein has also studied flows with closed orbits on 3-manifolds showing that they admit an $S^1$ action. By an averaging process the action is actually isometric and hence the flow is Riemannian  \cite[Remark 3]{carriere84}.

\subsection{Transverse Levi-Civita connection}

The {\em Bott connection} is a partial connection on the bundle $Q\to H$ defined by
\[
\mathring{\nabla}_{n} \bar X= \overline{ L_{n} X}.
\]
It is {\em partial} because the covariant derivative and transport are defined just along the leaves. This definition implies that the Bott parallel transport coincides with the slide along the leaf defined previously.

Assigned the transverse metric $g_T$, $L_n g_T=0$, there is a canonical {\em transverse Levi-Civita connection} for the vector bundle $Q\to H$, see \cite[Thm.\ 5.9]{tondeur97} and \cite[Sec.\ 3.3]{molino88} for a frame bundle formulation. In a simple cylindrical neighborhood $I
\times B$, as introduced above, the transverse Levi-Civita connection is nothing but the pullback connection of the Levi-Civita connection in $(B, g_T)$. In particular, in the direction of the leaves the corresponding transport coincides with the {\em sliding along the leaf}. In a more satisfactory non-local notation, the transverse Levi-Civita connection $\nabla^T$ is characterized by metric compatibility and absence of torsion \cite[Eq.\ (3.7)]{tondeur97} as follows  (here $X,Y$ are not necessarily foliate fields)
\begin{align*}
0&=(\nabla^T g_T)(\bar X,\bar Y):= \dd (g_T(X,Y))-g_T(\nabla^T  \bar X, \bar Y)-g_T( \bar X,\nabla^T  \bar Y), \\
0&=\nabla^T_X \bar Y-\nabla^T_Y \bar X-\overline{[X,Y]}
\end{align*}
where, once again, we do not change symbol to denote the scalar product on $Q_p$ (this should not cause confusion since the transverse metric is annihilated by $n$). The transverse Levi-Civita connection is  compatible with Bott's, namely {\em adapted}, $\nabla^T_n  \bar X=\overline{L_n X}$, and it is also {\em holonomy invariant} \cite[Thm.\ 5.11]{tondeur97} in the sense that  $L_n \nabla^T=0$ where we recall that \cite[Eq.\ 3.5]{tondeur97} (this object is tensorial  in $n, X,\bar Y$)
\[
(L_n\nabla^T)(X)\bar Y =\overline{L_n(\nabla^T_X \bar Y)}-\nabla^T_{L_nX} \bar Y-\nabla^T_X \overline{L_n Y}.
\]
This property could be expected as $L_n g_T=0$ and $\nabla^T$ is uniquely determined from $g_T$. A Koszul type formula for the transverse Levi-Civita connection can be provided \cite[Eq.\ (5.10)]{tondeur97} which can be  used as a definition. Alternatively, one can define a bundle-like metric $g_B$ and obtain the transverse Levi-Civita connection from the Levi-Civita connection $\nabla^B$ as follows \cite[Eq.\ (3.3)]{tondeur97}
\begin{align*}
\nabla^T_n \bar X&=\overline{L_n X},\\
\nabla^T_Y \bar X&=\overline{\nabla^B_Y X}, \qquad \textrm{for} \ Y: \ g_B(Y,n)=0.
\end{align*}
We described the transverse Levi-Civita connection because our aim is to clarify the connection between results in horizon geometry and Riemannian foliation theory. As previously observed, in horizon geometry we have already a connection $\nabla$ so it is natural to ask about the relationship between $\nabla^T$ and $\nabla$. The answer  helps translating results from one subject to the other.\\

\begin{proposition}
On horizon geometry we have $\nabla^T_Y \bar X= \overline{\nabla_Y X}$ where $\bar X \in \Gamma Q$, $X,Y\in
\Gamma TH$.
\end{proposition}

\begin{proof}
To start with, the right-hand side of the equation is well posed since if $X'=X+fn$ then $\nabla_Y X'-\nabla_Y X\propto n$ as $\nabla_Y n=\omega(Y) n$.
Since $g_T$ is the induced metric and $\nabla$ is the induced connection, they are compatible, thus the right-hand side defines a connection on $Q\to H$ which is metric compactible, indeed
\begin{align*}
&\dd (g_T(X,Y))-g_T(\overline{\nabla X}, \bar Y)-g_T( \bar X, \overline{\nabla Y})\\
&=\dd (g_T(X,Y))-g_T(\nabla X,  Y)-g_T( X,\nabla Y)=(\nabla g)(X,Y)=0
\end{align*}

It is also symmetric, again by the symmetry of $\nabla$,
\[
\overline{\nabla_X Y}-\overline{\nabla_Y X}-\overline{[X,Y]}=\overline{T^\nabla(X,Y)}=0.
\]
Thus, the right-hand side satisfies all the conditions for the transverse Levi-Civita connection and the latter is unique.
\end{proof}

The reader is referred to the book by Tondeur \cite{tondeur97} for other results on $\nabla^T$ that might  find application to horizons, and to the paper by Guti\'errez and Olea \cite{gutierrez16} for a related spacetime perspective.

\subsection{Isometric flows}

\begin{definition}
A flow  is {\em isometric} if  there is a Riemannian metric $h$ and a choice for the vector field $n$,  such that $L_n h=0$.\\
\end{definition}

\begin{remark}
When considering an isometric flow it is not restrictive to assume $h(n,n)=1$. It is sufficient to observe that $h'=h/h(n,n)$ satisfies $L_n h'=0$ and $h'(n,n)=1$. An isometric flow is thus represented by a unit Killing field over  a Riemannian space.\\
\end{remark}

Every isometric flow is a Riemannian flow, indeed, $g_T:=h- n^*\otimes n^*$, $n^*:=h(n,\cdot)$ has the desired properties of having signature $(0,+,\ldots,+)$, with kernel $\textrm{span}(n)$, and $L_n g_T=0$.

%
%

An important example of isometric flow is the following \cite[Example I.B.4]{carriere84}\\

\begin{example} \label{ort}
The sphere $S^3$ seen as a subset of $\mathbb{C}^2$, $\vert z_1\vert^2+ \vert z_2\vert^2=1$ with the metric induced from $\vert \dd z_1\vert^2+ \vert \dd z_2\vert^2$  admits the isometric flow $\varphi_t: (z_1,z_2) \mapsto (e^{i \lambda t} z_1, e^{i \mu t} z_2)$ where $\lambda$ and $\mu$ are non-zero real numbers with irrational ratio. There are just two closed orbits defined by $z_1=0$ for one, and $z_2=0$ for the other. The other orbits have as closure a torus of equation $\vert z_1\vert=a$, $0<a<1$. The lens space $L(p,q)$ is the quotient of $S^3$ under the action $(z_1,z_2) \mapsto (e^{2\pi i/p } z_1, e^{2\pi i/q} z_2)$ where $p$ and $q$ are non-zero integers.  The isometric flow $\varphi_t$ passes to an isometric flow on the quotient $L(p,q)$.\\
\end{example}

\begin{definition}
A flow  is {\em geodesible} if  there is a Riemannian metric $h$ and a choice for the vector field $n$,  such that $\nabla_n^h n=0$ and $h(n,n)=1$, where $\nabla^h$ is the Levi-Civita connection for $h$.\\
\end{definition}

In other words, the $h$-length of the orbits is stationary under small perturbations, i.e.\ the orbits of the flow are $n-1$-codimensional  minimal surfaces for the metric $h$ (the foliation is {\em taut}).

The following result improves \cite[Prop.\ III.B.1]{carriere84} \cite{molino85} in the last statement.\\

\begin{proposition} \label{kqf}
For a Riemannian flow on an $n$-dimensional manifold $H$ the following conditions are equivalent:
\begin{itemize}
\item[(i)] it is possible to choose $n$ and a field $P$ of transverse hyperplanes  that is left invariant by the flow of $n$ (i.e if $X\in P$ then $d\varphi_t(X) \in P)$,
\item[(ii)] it is possible to choose $n$ and a 1-form  $\sigma$ such that $\sigma(n)=1$, $L_n \sigma=0$,
\item[(iii)] it is isometric,
\item[(iv)] it is geodesible.
\end{itemize}
If $H$ is connected, compact and oriented, they are also equivalent to:
\begin{itemize}
\item[(v)] $H^{n-1}_b(\mathcal{F})\ne 0$.
\end{itemize}
Suppose (i)-(iv) hold, let $g_T$ be a preassigned Riemannian flow structure and let $P$ be as in (i). Then $n$ and a 1-form field $\sigma$ that satisfies (ii)  can be chosen so that  $n$ is a unit Killing field (hence geodesic) for the bundle-like metric $g_T+\sigma \otimes \sigma$ and $P=\ker \sigma$.\\
\end{proposition}

In short, the metrics (and forms) that realize the properties can all be chosen `at once' and compatibly with the Riemannian flow structure. This is important in our horizon interpretation, as we have assigned not only a Riemannian foliation, but also the bilinear form $g_T$.

It can be observed \cite{gluck79}  that (ii)  is equivalent to $\sigma(n)>0$ and $i_n \dd \sigma=0$, it is sufficient to rescale $n$.

\begin{proof}
$(i)\Rightarrow (ii)$.
The form $\sigma$ defined by $\ker \sigma=P$ and $\sigma(n)=1$ is such that $L_n \sigma=0$. This can be seen using $(L_n\sigma)(X)=L_n(\sigma(X))-\sigma(L_nX)$, and replacing $X=n$ or $X=e_i$ with $e_i$ basis for $P$ propagated in such a way that $L_n e_i=0$.

$(ii)\Rightarrow (iii)$. Let $g_T$ be such that $L_n g_T=0$ for some choice, and hence every choice, of $n$. In particular, the equation is true for the choice of point (ii).  Thus the metric $h:=g_T+\sigma \otimes \sigma$ is such that $L_nh=0$ (and such that $h(n,n)=1$).

$(iii)\Rightarrow (iv)$. Suppose $L_nh=0$ then $h':=h/h(n,n)$ is  such that $L_n h'=0$, $h'(n,n)=1$, and a  Killing vector field of unit norm is geodesic.

$(iv)\Rightarrow (i)$. This is a result by Sullivan \cite{sullivan79,gluck79}.
Let $h$ and $n$ be such that $\nabla^h_n n=0$ and $h(n,n)=1$. Let $P=\ker h(n, \cdot)$, note that $P$ is transverse to $n$. Let $X$ be a vector field such that $L_n X=0$
\begin{align*}
L_n (h(n,X))&=\nabla_n^h (h(n,X))=h(\nabla_n^h n, X)+h(n, \nabla_n^h X)=h(n,\nabla_X^h n)\\
&=\frac{1}{2}   X (h(n,n))=0,
\end{align*}
which proves that $P$ is preserved.

The equivalence $(iii) \Leftrightarrow (v)$ is due to Molino and Sergiescu \cite{molino85}.

For the last statement just concatenate the first two implications above.
\end{proof}

\begin{remark} \label{pem}
Let us consider a Riemannian flow and let $g_B$ be a bundle-like metric. If we are in the isometric case $n$ and $g_B$ can be chosen such that, $L_n g_B=0$, and we can rescale $g_B$ so that, without loss of generality, $g_B(n,n)=1$. Let $\chi=g_B(n,\cdot)$ be the characteristic form,  then, as $\chi$ is determined by $g_B$, $L_n \chi=0$, and so $\dd \chi$ is basic.

In the more general Riemannian flow case, let $n$ be such that $g_B(n,n)=1$ and let $\chi=g_B(n,\cdot)$, so that $g_B=g_T+\chi \otimes \chi$. We define the {\em mean curvature 1-form} $\mu:=L_n \chi=i_n \dd \chi$, which has the property $i_n \mu=0$, and the  {\em Euler 2-form}
$\mathfrak{e}:= \dd \chi-\chi \wedge \mu$ which has also the property $i_n \mathfrak{e}=0$. Clearly, these forms depend only on $\chi$ and not on $g_T$, where $\chi$ is a 1-form positive on the flow (so $n$ can be normalized from $\chi$).

The terminology for $\mu$ follows from the following calculation with $X\in TH$
\begin{align*}
g_B(\nabla^B_n n,X)&=  {n (g_B(n,X))}-g_B(n, \nabla^B_n X)=\mu(X)+g_B(n, L_nX)-g_B(n, \nabla^B_n X)\\&=\mu(X)+g_B(n,\nabla^B_X n)=\mu(X),
\end{align*}
which can also be written
\[
\nabla^B_n n=\mu^{\sharp_B}.
\]
That is, $\mu$ is the 1-form dual to the acceleration.

 Suppose that $H$ is connected and compact. A deep theorem of foliation theory \cite{dominguez98}  states that $\mathfrak{e}$ and $\mu$ are basic for some choice of $g_B$ (or better $ \chi$) and in this case $\mu$ is closed (actually we have always $\dd\mu_b=0$ for the basic component of the mean curvature form). The {\em {\'A}lvarez class} $[\mu]\in H^1_b(\mathcal{F})$ is important as it depends only on the foliation. We have that the foliation is isometric iff $[\mu]=0$, see \cite[Thm.\ 6.4]{alvarez92}, a result which provides a sixth characterization of isometric flows. The Euler class is also important and clarifies the relationship between cohomology and basic cohomology \cite{royo01}.\\
\end{remark}

Another result by Molino and Sergiescu \cite{molino85} which is worth recalling states that a Riemannian flow on a connected compact manifold admits a cross-section (i.e.\ it comes from a suspension) if and only if the natural map $H^{n-1}_b(\mathcal{F})\to H^{n-1}(\mathcal{F})$ is non-zero. In other words if the foliation comes from a suspension there is a basic $n-1$-form which is not the differential of a, not necessarily basic, $n-2$-form. In the suspension case this form exists and is provided by the ($n-1$-)volume of the transversal lifted to $H$.

If $(H,\mathcal{F})$ is a suspension then its Riemannian flow is geodesible and isometric \cite{gluck79}. Indeed, consider the manifold $[0,1]_t\times T$, where $H$ is obtained identifying $(1,p)$ with $(0, \phi^{-1}(p))$, then the metric $\dd t^2 \oplus g_T$, with $g_T$ Riemannian metric on $T$, is well defined also after the identification,  $n=\p/\p t$ is unit Killing for it, thus its integral curves are geodesics.

Under the dominant energy condition $L_n \omega=\dd \kappa$ \cite[Lemma 6]{minguzzi21} thus, taking into account that non-degenerate compact horizons have a surface gravity that can be normalized to a non-zero constant  (and conversely) \cite{petersen18b,minguzzi21,reiris21}, we get\\

\begin{corollary} \label{mqp}
Under the dominant energy condition  the Riemannian flow on non-degenerate compact horizons is isometric.
\end{corollary}

\begin{proof}
(first proof) Let $n$ be such that $\kappa \ne 0$ is constant. Let $\sigma:= \omega/\kappa$, then $L_n \sigma=0$, $\sigma(n)=1$, and property (ii) of Prop.\ \ref{kqf} is satisfied.

(second proof) Let $n$ be such that $\kappa \ne 0$ is constant. Let $\chi:= \omega/\kappa$, then $\mu=i_n \dd \chi=i_n \dd \omega/k=0$, hence $[\mu]=0$ which by Alvarez's result implies that the flow isometric.
\end{proof}

\begin{remark}
It is not true that  under the dominant energy condition the compact horizons whose Riemannian flow is isometric are necessarily non-degenerate. For instance, consider  an extremal (Kerr or Kerr-Newmann) black hole and compactify the horizon with the usual trick \cite{friedrich99}. The horizon is a trivial $S^1$-bundle (e.g.\ \cite{minguzzi24}) hence a suspension, hence isometric. More simply, one can obtain a stack of degenerate horizons from the flat metric $\dd u \dd x+\dd y^2+\dd z^2$ with the identifications $x+1\sim x$, $y+1\sim y$, $z+1\sim z$. The obtained horizons $u=const.$ admit the bundle-like metric $\dd x^2+\dd y^2+\dd z^2$ invariant under the Killing field $\p/\p x$ and so the Riemannian flows are isometric. Yet the lightlike geodesics running over them are clearly complete.\\
\end{remark}

The previous result Cor.\ \ref{mqp} was noted first in the vacuum case by Petersen \cite{petersen18b} and used by Bustamante and Reiris \cite{bustamante21} to classify non-degenerate horizons. Indeed, one can take advantage of the Myers-Stenrood theorem, which states that the isometry group of a smooth compact Riemannian
manifold is a compact Lie group. The full machinery of Lie group theory can then be applied \cite{isenberg92,alexandrino15}.

We note that it is possible to infer that the Riemannian flow is isometric without any non-degeneracy assumption by imposing the simple connectedness of the horizon. This fact follows from a result on Riemannian foliation theory due to Ghys \cite{ghys84} \cite[App.\ B, Cor.\ 2.5]{molino88}\\

\begin{proposition} \label{kthf}
On a compact simply connected $n$-dimensional manifold every Riemannian flow is isometric.\\
\end{proposition}

Since the flow of any timelike vector field can be used to establish a homeomorphism between a partial Cauchy hypersurface and the horizon of the (future/past) Cauchy development, we conclude\\
\begin{corollary} \label{clf}
Suppose that a spacetime satisfies the null energy condition and admits a compact simply connected partial Cauchy hypersurface $S$. Then the flow on the horizon $H^+(S)$ (and $H^-(S)$) is isometric. (In spacetime dimension 4 the horizon has topology $S^3$ by Perelman's theorem.)\\
\end{corollary}


In general, isometric flows are more restrictive than Riemannian flows and without imposing non-degeneracy conditions or a simply connectedness assumption one has to deal with the latter objects. Luckily, Carri\'ere (1984) sought to classify Riemannian flows on compact 3-manifolds precisely to generalize the classification on isometric flows that was already clear at the time. We apply these  stronger classification results to horizons in the next section. It can be said that while Bustamante and Reiris classification of non-degenerate horizons \cite{bustamante21} is based on that for isometric flows, the below classification of general horizons is based on that for Riemannian flows.

\subsection{Structure theorems}

In this section we assume that $(H,g_T)$ is a Riemannian flow (oriented 1-di\-men\-sional Riemannian foliation). Of course, we are interested in the case in which this structure comes from that of a compact horizon as this will lead to previously unknown results in the horizon case.

We start with
\cite[Prop.\ 4, Cor.\ 5]{carriere84}:\\

\begin{theorem}
The minimal volume in Gromov's sense of a compact $H$  is zero, and so $H$ does not admit a metric with negative sectional curvature.\\
\end{theorem}

A similar result was obtained in the context of horizons by Rendall \cite{rendall98} who used the notion of {\em  Riemannian manifold collapsing with bounded diameter}. We note that a similar proof and calculation appeared previously in  \cite{carriere84b}.

We say that the foliation $\mathcal{F}$ has   {\em regular closure} if the closures of the leaves have all the same dimension.
Of the structure theorems by Molino we recall \cite[Lemma 5.2]{molino88}\\

\begin{proposition}
If $(H,\mathcal{F}, g_T)$ is a Riemannian flow with regular closure on a compact connected manifold, then the foliation given by the closures of the leaves $\overline{\mathcal{F}}$ is Riemannian and every metric which is bundle-like for $\mathcal{F}$ is also bundle-like for  $\overline{F}$.\\
\end{proposition}

\noindent and \cite[Prop.\ 5.2]{molino88}\\

\begin{proposition}
If $(H,\mathcal{F}, g_T)$ is a Riemannian flow with regular closure on a compact connected manifold  then the space of closures of the leaves $H/\overline{\mathcal{F}}$  has a Satake manifold structure for which the projection $H\to H/ \overline{\mathcal{F}}$ is a morphism of Satake manifolds.\\
\end{proposition}

Other structure theorems by Molino for general Riemannian foliations were subsequently specialized and improved by Carri\`ere in the flow case \cite{carriere84,carriere84b}\cite[App.\ A]{molino88}, so we shall use directly Carri\`ere's versions.

We have \cite[App.\ A, Thm.\ 1.1]{molino88} \cite[Thm.\ A]{carriere84b}\\

\begin{theorem} \label{nqpx}
In a Riemannian flow over a compact manifold $H$ the closures of the leaves of $H$ partition $H$ into  embedded submanifolds that are diffeomorphic to tori. The foliation induced on each torus is differentiably conjugate (without parameter) to a linear flow on the torus with dense orbits.\\
\end{theorem}

 In particular, if a generator densely fills the horizon, then $H$ is a $n$-torus $\mathbb{T}^n$ and  the flow is differentiably conjugate to a linear flow on $\mathbb{T}^n$.

The physical spacetime is 4-dimensional so for us the most interesting horizon case is obtained for $n=3$. We now impose this dimensionality condition.

In the following result $k+1$ is the dimension of the closures $\bar L\simeq \mathbb{T}^{k+1}$, $L
\in \mathcal{F}$, of maximal dimension (they are tori by the previous theorem).

Observe that if the spacetime is orientable then the horizon is orientable. Otherwise, it is possible to obtain a classification result by passing to a double covering.

The next theorem follows  from \cite[Thm.\ III.A.1]{carriere84} \cite[Cor.\ III.B.4]{carriere84}  (the more detailed version we give comes from looking at the original proof)\\

\begin{theorem} \label{car}
Let $(M,g)$ be a  smooth, $4$-dimensional spacetime and  let $H$ be  an oriented compact horizon in $M$. There are the following mutually excluding possibilities:
\begin{enumerate}
\item[(i)] All orbits are closed, $k=0$. $H$ is a Seifert fibration over a Satake manifold. The  fibers are the orbits of the flow.
\item[(ii)] The flow has precisely two closed orbits and every other orbit densely fills a two-torus, $k=1$. If the closed orbits are removed the manifold is diffeomorphic to a product $(0,1) \times  \mathbb{T}^2$ where the flow in each fiber $\{s\} \times \mathbb{T}^2$ is a linear flow with dense orbits. The manifold $H$ is obtained by gluing two solid tori by their boundaries, on each solid torus the foliation being given through a suspension by an irrational rotation of the disk. Moreover, there are two subcases:
\begin{itemize}
\item[a)] $H$ is diffeomorphic to the lens space\footnote{It comprises the case $L(1,0)=S^3$ while the case $L(0,1)=S^1 \times S^2$ is included in (ii-b). See \cite{saveliev12} for this gluing construction. Note that a $S^3$ topology can also appear in case (i) (e.g.\ Hopf bundle).} $L(p,q)$,  and the flow is conjugate to that of the example \ref{ort}.
\item[b)] $H$ is diffeomorphic to $S^1\times S^2$ and the flow is conjugate to the flow given by the suspension of an irrational rotation of $S^2$ (with respect to, say, the $z$-axis in the standard isometric embedding in $\mathbb{R}^3$).
\end{itemize}
\item[(iii)] Every orbit densely fills a two-torus, so no orbit is closed, no orbit is dense,  $k=1$. $H$ is a $\mathbb{T}^2$-bundle over $S^1$, each  torus fiber being the closure of the orbits in it, and there are two possibilities:
\begin{itemize}
\item[a)] $H$ is diffeomorphic to $\mathbb{T}^3$ with a flow conjugate to a linear flow on the torus.
\item[b)] $H$ is diffeomorphic to the hyperbolic fibration $\mathbb{T}^3_A$, ($
\textrm{tr} A > 2$) and the flow is conjugate to one of the flows of the next example \ref{exm}.
\end{itemize}
\item[(iv)] All orbits are dense, $k=2$. $H$ is diffeomorphic to $\mathbb{T}^3$ with a flow conjugate to a linear flow on the torus.
\end{enumerate}
Moreover, except case (iii-b), the flow is isometric.\\
\end{theorem}

In all instances, except  case (iii-b), $H$ is a Seifert manifold  \cite[App.\ A, Thm.\ 4.4]{molino88} (the fibration that makes it a Seifert manifold is not necessarily the lightlike generators foliation detailed in the previous list). The horizon can be simply connected, and hence $S^3$ (Perelman's theorem), only in cases (i) and (ii-a).

Let us check which of these flows are isometric.
Case (i) is isometric as it is possible to start from an arbitrary Riemannian metric and average by using the $S^1$ action to get an invariant metric \cite[Example 3]{carriere84}.
Case (ii-a) is isometric, the invariant metric being that projected from the $S^3$ covering. Case (ii-b) is isometric as the Riemannian flow comes from a  suspension.  Cases (iv) and (iii-a) are isometric, where the invariant metric is the flat metric on $\mathbb{T}^3$.

The case (iii-b) is left out, in fact, as mentioned by the theorem, Carri\`ere proved that it is not isometric \cite[Cor.\ III.B.4]{carriere84}, and hence that it is the only possible non-isometric Riemannian flow on a compact oriented 3-dimensional manifold.

Joining Cor.\ \ref{mqp} with Thm.\ \ref{car} we obtain\\

\begin{corollary} \label{vmx}
In a 4-dimensional spacetime and under the dominant energy condition an oriented non-de\-ge\-ne\-ra\-te compact horizon $H$ and its flow has the  possible structures of Thm.\ \ref{car} save for case (iii-b) that does not apply, that is
\begin{itemize}
\item[i)]  if all the generators are closed  then $H$ is a Seifert manifold,
\item[ii)] if only two generators are closed and every other generator densely fills a two-torus then $H$ is a lens space,
\item[iii)] if every generator densely fills a two-torus, then $H$ is a $\mathbb{T}^2$-bundle over $S^1$ (and diffeomorphic to $\mathbb{T}^3$),
\item[iv)] If every generator densely fills the horizon, then $H$ is a three-torus $\mathbb{T}^3$.
\end{itemize}

\end{corollary}

This is precisely Bustamante and Reiris' classification of non-degenerate horizons \cite{bustamante21} (provided orientability is assumed) improved, as we already observed in \cite{minguzzi21}, replacing the vacuum condition with the dominant energy condition.

A first consequence of our study of Riemannian flows in connection to horizons is then\\

\begin{theorem}
In 4 spacetime dimensions an  oriented compact horizon $H$   has the same possible structures previously determined for  a non-degenerate one under the dominant energy condition, save for the additional possibility (iii-b).\\
\end{theorem}

We do not make here any assumption on the existence or non-existence of complete generators on $H$, and we do not impose any energy condition, save possibly for the null energy condition. The point is that the structure theorem \ref{car} depends only on the fact that the generators flow is Riemannian, and this property uses only the smoothness of the horizon and its totally geodesic property which follow from the null energy condition.\\

\begin{remark}
A similar classification result for compact horizons in 5-dimen\-sio\-nal spacetime could be straightforwardly formulated by using the results by Almeida and Molino \cite{almeida86}. In such dimension there are many possible non-isometric structures.\\
\end{remark}


\begin{example} \label{exm}
We present an example of totally geodesic compact lightlike hypersurface which displays a non-isometric Riemannian flow dynamics. Unfortunately,  our example does not satisfy the null energy condition. We do not know if a more physically reasonable example exists (the difficulty in finding one confirms a recent conjecture by Isenberg-Moncrief, see below). The horizon dynamics is the prototypal one of Theorem \ref{car} case (iii-b).

Every matrix $A\in SL(2,\mathbb{Z})$ has an inverse which has also integer entries and so belongs to the same space. Consider the map $A: \mathbb{R}^2 \to \mathbb{R}^2$. Altering the components of the entry vector by  integers results in the components of the image vector being altered by integers. As a consequence, it passes to a bijective map, which we denote in the same way, $A: \mathbb{T}^2 \to \mathbb{T}^2$, $\mathbb{T}=\mathbb{R}/\mathbb{Z}$.
Let $\textrm{tr}(A)>2$. The characteristic polynomial is $p_A(\lambda)=\lambda^2-\textrm{tr}(A) \lambda +1$ so $\Delta=\sqrt{\textrm{tr}(A)^2-4}>0$.

This discriminant is actually irrational. Indeed, if it is rational it is of the form $a/b$ with $a$ and $b$ positive coprime integers, but squaring it we get $\textrm{tr}(A)^2-4 = \frac{a^2}{b^2}.$ Since $tr(A)^2-4$ is an integer, it follows that $b^2$ is a divisor of $a^2$ and hence $b = 1$ since $a$ and $b$ are positive coprime integers. It follows that $(\textrm{tr}(A)- a)(\textrm{tr}(A)+a) = 4$ which gives the contradiction using the fact that $\textrm{tr}(A)>2$. We conclude that $\Delta$ is irrational.

As the two  roots  $\lambda_{\pm}=({\textrm{tr}(A)}\pm \Delta)/2$ are distinct (and irrational)  the matrix is diagonalizable. Observe that by $\det A=1$, $\lambda_+ \lambda_-=1$, so $0<\lambda_-<1<\lambda_+$. The map $A: \mathbb{T}^2 \to \mathbb{T}^2$ is the canonical example of linear Anosov diffeomorphism. 

Let $v_+$ and $v_-$ be the real eigenvectors. Writing $v_+= (a, b)^T$ either $a$ or $b$ is different from zero. In the former case we consider the  first line of $A(v_+)=\lambda_+ v_+$, while in the latter case we consider the second line. Either way, we find $a,b\ne 0$ with an irrational ratio $a/b$. A similar conclusion is reached considering $v_-$.

On the torus $\mathbb{T}^2$ we have now two foliations. That of direction $v_+$ and that of direction $v_-$. All orbits of any of these foliations are
dense in $\mathbb{T}^2$ due to the irrational ratio considered in the previous paragraph. Moreover, $A$ respects the foliations in sending the leaves of the plus foliation to leaves of the plus foliation (if $v'-v \propto v_+$ then $A(v')-A(v)\propto v_+$), and similarly for the minus foliation.

We construct the suspension $\mathbb{T}^3_A=\mathbb{R} \times \mathbb{T}^2/\!\!
\sim$, $(x, v) \sim (x+1, A(v))$, which is thus a $\mathbb{T}^2$ bundle over $S^1$. This gluing respects the plus and minus foliations.

The 1-dimensional integral curves of $X:=\p/\p x$  should not be confused with the Riemannian foliation orbits (lightlike generators) that we shall introduce next.

The example (iii-b) is $\mathbb{T}^3_A$ endowed with the 1-dimensional (suitably oriented) foliation given by the plus or minus foliations. Carri\`ere proved that both are Riemannian but non-isometric.

We focus on one of these oriented foliations, call $\lambda$ the eigenvalue
and $v$ the corresponding eigenvector of $A$. Then $Y=\lambda^{x} v$ passes to the quotient $A_*(Y)(p)=\lambda^{x} \lambda v=\lambda^{x+1} v=Y(A(p))$. Let $\lambda^{-1}$ be the other eigenvalue
and $w$ the corresponding eigenvector of $A$. Then $Z= \lambda^{-x} w$ passes to the quotient.
On the $\mathbb{R}^3$ covering we introduce coordinates $(x,y,z)$ so that $v=\frac{\p}{\p y}$, $w=\frac{\p}{\p z}$, and so $Y=\lambda^{x}\frac{\p}{\p y}$, $Z=\lambda^{-x}  \frac{\p}{\p z}$. The dual basis to $X,Y,Z$ is $\dd x, \lambda^{-x} \dd y,  \lambda^{x} \dd z$.

The fields $X,Y,Z$ satisfy the following commutation relations (the constant $\log \lambda$ could be sent to 1 redefining $X$) \cite{ghys01}
\[
[X,Y]=(\log \lambda) Y, \qquad [X,Z]=-(\log \lambda) Z, \qquad [Y,Z]=0.
\]
They define a Lie algebra and the connected Lie group associated to it is often denoted $G_3$.

Consider the bilinear forms
\begin{align*}
g_T&= \dd x^2+ (\lambda^{x} \dd z)^2, \\
g_B&= \dd x^2+ (\lambda^{-x} \dd y)^2 + (\lambda^{x} \dd z)^2.
\end{align*}
Where $g_B$ is a bundle-like metric for the flow determined by $Y$ (it is immediate that $L_Y g_T=0$). In our previous notation $n:=Y$, $n$ is $g_B$-normalized, and the characteristic form is  $\sigma=\lambda^{-x} \dd y$. Observe that the mean curvature 1-form is
\[
\mu=L_Y \sigma=\lambda^{-x} \dd (L_Yy)=\lambda^{-x}  \dd \lambda^x=(\log \lambda) \dd x,
\]
which, being annihilated by $Y$ and independent of $y$, is basic. Observe that it is closed but it is not exact due to the identification that defines $\mathbb{T}^3_A$. Thus $[\mu]\ne 0$ and hence the Riemannian foliation is not isometric by \`Alvarez's result, see Remark \ref{pem}.
We note that the Euler class vanishes $\mathfrak{e}=\dd \sigma -\sigma \wedge \mu=0$.

Having described the Riemannian flow structure of case (iii-b), we need to clarify if such geometry can be realized as a horizon.

Let us consider the spacetime metric on $\mathbb{R}_u\times \mathbb{T}^3_A$
\begin{align} \label{cqp}
g&= -2(\lambda^{-x} \dd y)(\dd u+ f(u)\lambda^{-x}\dd y) +g_T \nonumber\\
&=-2\lambda^{-x} \dd y (\dd u+ f(u)\lambda^{-x}\dd y) +\dd x^2+ \lambda^{2x} \dd z^2
\end{align}
 with the time orientation given by the lightlike vector $N=\p/\p u$. Suppose that $f(0)=0$ and let $H=\{u=0\}$.

 The metric induced on $H$ is indeed $g_T$. A future-directed lightlike vector field on $H$ is given by $n=Y=\lambda^{x}\frac{\p}{\p y}$ and this expression defines $n$ also outside $H$ in such a way that $[n,N]=0$. Since
\[
g(n,n)=-2 f,
\]
the vector field $n$ is timelike on the region where $f>0$ and spacelike on the region where $f<0$.
On $H$ we have  $\sigma=-g(N,\cdot)$. By Prop.\ \ref{bort} $H$ is totally geodesic.

By  \cite[Eq.\ (47)]{minguzzi21} the surface gravity is
\[
\kappa=\frac{1}{2}\frac{\p}{\p u} g(n,n)=- f'(0),
\]
which is constant throughout $H$.
It can be observed that $\frac{\p}{\p y}$ is a local Killing field for (\ref{cqp}) but it does not globalize. The global field is $n$ but, using $L_n \sigma=\mu=(\log \lambda) \dd x$
\begin{align*}
L_n g
&=- (\log \lambda) [\dd x \otimes \dd u+\dd u \otimes \dd x +2f(u)(\dd x \otimes \sigma+\sigma \otimes \dd x)]
\end{align*}
thus $n$ is not Killing. Similarly, $\frac{\p}{\p z}$  is a local Killing field and $\overline{\frac{\p}{\p z}}$ is a transverse local Killing field (see Sec.\ \ref{killv}).

For $X\in TH$, $\nabla_X n=\omega(X) n$ where $\omega$ can be obtained from the Christoffel symbols calculated with the help of Mathematica
\[
\omega=\frac{1}{2} \log \lambda\dd x -f'(0) \sigma,
\]
which confirms $\kappa=\omega(n)=-f'(0)$.

The covariant Ricci tensor is
\begin{align*}
Ric&=-\frac{3}{2} (\log \lambda)^2 \dd x^2 + 2   (\lambda^{-x} \dd y)( f''(u)\dd u-f'(u)\log \lambda \, \dd x)+f(u) [(\log \lambda)^2\\
&\quad +2 f''(u)](\lambda^{-x} \dd y)^2 ,
\end{align*}
the scalar curvature is
\begin{align*}
S&=-\frac{3}{2} (\log \lambda)^2 -2  f''(u) ,
\end{align*}
the covariant Einstein tensor is
\begin{align*}
G&=[-\frac{3}{4}(\log \lambda)^2+ f''(u)] \dd x^2 - 2   (\lambda^{-x} \dd y)( \frac{3}{4}(\log \lambda)^2 \dd u+f'(u)\log \lambda \, \dd x)\\
& \quad -\frac{1}{2} f(u) (\log \lambda)^2(\lambda^{-x} \dd y)^2 + \frac{1}{4} (3 (\log \lambda)^2 +4f''(u))\lambda^{2x} \dd z^2.
\end{align*}

By using these expressions it is possible to show that  for every choice of function $f$ the null energy condition is not satisfied. This result supports  the conjecture by Isenberg and Moncrief that degenerate Cauchy horizons do not exist in the analytic case \cite{moncrief20} (see also \cite{chrusciel06,khuri18} for similar questions in the black hole case).

\end{example}


\subsection{More on isometric flows: Transverse Killing fields} \label{killv}

A recurring problem in spacetime horizon geometry is that of establishing the existence of further Killing vector fields beyond the one that might be aligned with the generators of the horizon. In the context of horizons that come from the compactification of black hole horizons the existence of such symmetry is often referred as {\em axisymmetry} \cite{hollands07,hollands09}. It is natural to look for the existence of Killing fields for the Riemannian flow  structure $(H,g_T)$, as hopefully the extension of the Killing field in a neighborhood of the horizon could be accomplished with PDE methods.

Fortunately, the theory of Riemannian flows provides  rather accurate answers to this type of questions.
Let us consider a Riemannian flow $(H,g_T)$.
A vector field $X \in \mathfrak{X}(H)$ is a  {\em Killing  field for the transverse metric} if $L_X g_T=0$. Every Killing  field for the transverse metric is a foliate field, $[X,n]
\propto n$, see \cite[Lemma 3.5]{molino88}, indeed if $Z\in \mathfrak{X}(H)$, as $[X,Z]\in \mathfrak{X}(H)$
\begin{align*}
0&=(L_Xg_T)(n,Z)=X(g_T(n,Z))-g_T(L_Xn, Z)-g_T(n,L_X Z)=-g_T(L_Xn, Z)
\end{align*}
and from the arbitrariness of $Z$ it follows $L_Xn\propto n$, i.e.
 $X$ is a foliate field.

Observe that if $X$ is a Killing  field for the transverse metric then $X+f n$ is a Killing field for the transverse metric too, where $f$ is any function. Indeed, for $Y,Z\in  \mathfrak{X}(H)$
\begin{align*}
L_{X+fn} g_T(Y,Z)\!&=(X\!+\!fn) (g_T(Y,Z))-g_T([X\!+\!fn, Y], Z)-g_T(Y,[X\!+\!fn,Z])\\
&=L_{X} g_T(Y,Z)+f L_{n} g_T(Y,Z) =0
\end{align*}
The projection $\bar X$ of $X$ is called {\em transverse Killing  field}. By the previous observation, any representative $X$ of a transverse Killing  field  $\bar X$ is a  Killing  field for the transverse metric. The Killing fields for the transverse metric form a Lie algebra, and it is easily checked that for $X,Y$ representatives of transverse Killing fields, $[\bar{X},\bar{Y}]:=\overline{[X,Y]}$ does not depend on the representatives, which shows that transverse Killing fields form a Lie algebra too.

In the following result $k+1$ is the dimension of the closures $\bar L\simeq \mathbb{T}^{k+1}$, $L
\in \mathcal{F}$, of maximal dimension (they are tori by the  Theorem \ref{nqpx}).

There is a further characterization of isometric flows that we did not mention so far and that reads as follows: the central transverse sheaf is globally trivial \cite{molino85}. Without entering into the details of this strong result we  summarize just one direction which is of interest to us (it can be read as saying that non-degenerate horizons are isometric foliations and hence Killing foliations \cite[Example 5.4]{alexandrino22}, see \cite{alexandrino22} for the latter notion)\\

\begin{theorem} \label{cnpz}
Let $H$ be compact and oriented. If $(H,g_T)$ is an isometric flow (e.g.\ because it comes from a non-degenerate horizon and the dominant energy condition holds, cf.\ Cor.\ \ref{mqp}) or $H$ is simply connected, then there are $k$ Killing fields $X_i$, $i=1,\cdots, k$, for  (the transverse metric) $g_T$  such that $[X_i,X_j]\propto n$ and at each point $\textrm{span}(n,X_1,\ldots, X_k)$ is the tangent space to a tori which is the closure of a generator passing through that point. Moreover, for each $i$, $[X_i,Y]\propto n$, where $Y$ is any foliate vector field.\\
\end{theorem}

If the isometric assumption is not satisfied the Killing fields exist locally \cite{molino85} but certainly some of them does not globalize \cite{molino85}, e.g.\ the field $\p/\p z$ in example \ref{exm}.

If the generators do not close then $k\ge 1$ and so there exists a (non-vanishing) transverse Killing field. Similar results in the spacetime framework were obtained in \cite{isenberg92,hollands07}.

The vector fields $X_i$ will in general be linearly dependent at some points. Mozgawa's theorem clarifies to what extent the fields can be found to be linearly independent at every point \cite{mozgawa85}\cite[Thm.\ 3.6]{alexandrino22}. The result is that there are at least $r$ Killing fields $\bar X_i$, linearly independent at every point, where $r+1$ is the smallest dimension for a generator closure.

Observe that in the Thm.\ \ref{cnpz} one gets just $[X_i, n] \propto n$ which is a bit unsatisfactory. However, we can improve it as follows\\

\begin{proposition}
Let $H$ be an oriented compact non-degenerate horizon and suppose that the dominant energy condition holds, then the Killing fields $X_i$ in Thm.\ \ref{cnpz} can be chosen so that $[X_i,n]=0$ and $\omega(X_i)=0$  for every $i$. 
\end{proposition}

\begin{proof}
Indeed, we know  \cite{reiris21,minguzzi21}  that $n$ can be chosen so  that   $\omega(n)=\kappa \ne 0$ where $\kappa$ is a constant. We can replace $X_i\to X_i +h_i n$, so that $\omega(X_i)=0$. Let $[X_i,n]=f_i n$. Under the dominant energy condition \cite[Lemma 5]{minguzzi21} $i_n\dd\omega=0$, thus
\[
0=\dd \omega(n,X_i)=n(\omega(X_i))-X_i(\omega(n))-\omega([n,X_i])=f_i \kappa
\]
from which we get $f_i=0$.
\end{proof}

\begin{corollary}
Let $H$ be an oriented compact non-degenerate horizon and suppose that the dominant energy condition holds and that some generator does not close,\footnote{In spacetime dimension 4 this is equivalent to the request that $(H,\mathcal{F})$ is not a Seifert fibration \cite{epstein72}.} then there is a vector field $X\in \mathfrak{X}(H)$ not everywhere aligned with $n$, such that $L_X g_T=L_n g_T=0$, $[X,n]=0$.
\end{corollary}


\begin{proof}
By Thm.\ \ref{nqpx}  the closure of the non closed generator is a torus $\mathbb{T}^{s+1}$, $s\ge 1$, hence $k\ge 1$.
\end{proof}

The previous and the next results do not rely on analyticity or vacuum assumptions and hold in any dimensions. They can be compared with results in \cite{hollands07,moncrief20} where the  authors imposed stronger assumptions but the Killing field was found in a (one-sided) neighborhood of the horizon.\\

\begin{corollary}
On a strongly causal spacetime that satisfies the dominant energy condition, let $H^+$ be a  totally geodesic smooth connected component of a  stationary black hole event horizon and suppose that $H^+$ admits an oriented compact cross-section (these assumptions are natural, see \cite{minguzzi24}). Suppose that $H^+$ is non-degenerate and that somewhere on $H$ the  Killing vector $k$ does not align with $n$ (non-static, i.e.\ rotating).
Then there is a vector field $X\in \mathfrak{X}(H)$ not everywhere aligned with $n$, such that $L_X g_T=L_n g_T=0$, $[X,n]=0$.
\end{corollary}

\begin{proof}
It is sufficient to compactify the horizon via a standard trick \cite{friedrich99,moncrief08}. The period $\tau$ of the Killing flow used for the identification must be chosen so that at least one generator does not close. To accomplish this result, first consider the quotient $\tilde S$ of $H^+$ with respect to the geodesic (generator) flow (existence and triviality of all the bundles used follows from results in \cite{minguzzi24}). Since the flow of $k$ sends generators to generators it satisfies $[k,n]\propto n$ and projects on a Killing field $\bar k$ on the quotient (endowed with a Riemannian metric whose pullback is $g_T$).

Let us pick a lightlike generator where $k$ is not aligned to it and consider its projection $\bar p$. If the orbit of $\bar k$ starting from $\bar p$ is not closed the period $\tau$ can be chosen arbitrarily, the generator is not going to close in the compactified space. If instead, for the special choice $\bar \tau$ it closes, we choose $\tau=r \bar \tau$ with $r$ irrational, so that again, the  generator is not going to close in the compactified space.
\end{proof}

These theorems should be compared with those obtained by using the theory of Lie groups \cite{isenberg92,alexandrino15,bustamante21} which can be applied in the non-degenerate case as the flow is isometric.

We stress that thanks to the connection with the theory of Riemannian flows,  some of the above results hold also in the degenerate case, with the modification that $X$ (or the fields $X_i$) exists just locally, in a neighborhood of every point. Fields from different patches cannot  just be matched together at the intersections in a consistent global way.\\

\begin{remark}
We mention that the previous results were meant to help  establish the existence of symmetries for the horizon. Such symmetries might in turn constrain  the dynamical structure. For instance, in a 4-dimensional spacetime, if the spacetime admits a conformal Killing field tangent to the  horizon and everywhere transverse to the generators, then the horizon cannot have the structure  (ii) of Thm.\ \ref{car}. Indeed, the closed generators would be mapped to distinct arbitrarily close closed generators, which is not possible as the two closed generators are isolated.
\end{remark}

%
%

\section{Surface gravity}
One would like to prove that on a degenerate horizon the vector field $n$ can be chosen so as to have zero surface gravity. Ultimately, this might help clarifying the geometry of such horizons, possibly leading to the proof that  degenerate Cauchy horizons do not exist. We start investigating the necessary conditions for  constant  surface gravity.


\subsection{Necessary conditions for  constant surface gravity. }

Let us look for necessary conditions for the existence of a lightlike field $n$ of constant surface gravity. Observe that we do not impose dimensionality conditions on the horizon.\\

\begin{lemma}
\label{nc}
Let $n$ be a smooth tangent vector field of constant zero surface gravity.
For any smooth function $f$ on $H$, and $n^{\prime} = e^fn$, the following holds,
\begin{itemize}
\item[(i)]
\begin{equation}
\label{equ1}
\omega^{\prime}(n^{\prime}) = \kappa^{\prime} =  n^{\prime}(f).
\end{equation}
\item[(ii)] For any integral curve $\gamma$ of $n^{\prime}$ such that $\gamma(0)= p, \gamma (s) = q, \\
 s \in [0, +\infty)$; it holds $ \int_{\gamma([0, s])}\omega^{\prime} = f(q) - f(p).$ \\
In particular, for any periodic integral curve we get: $ \int_{\gamma}\omega^{\prime} = 0$.
\end{itemize}
\end{lemma}

\begin{proof}
It is  immediate from Eq.\ (\ref{ren1}).
\end{proof}

Point (ii) implies that up to a constant $f$ can be uniquely determined on $\bar \gamma$ first determining it on the curve and then by using continuity. By Carri\`ere's result $\bar \gamma$ is a torus.\\

\begin{proposition}    \label{selc}
Let $H$ be a compact horizon. Let $n'$ be a future directed lightlike tangent field and suppose that there is another choice $n$, $n'=e^f n$, such that $n$ has constant surface gravity $\kappa$.
Let $\alpha:H\to  T^*H$ be a 1-form  such that $\alpha(n')=1$ and let us consider the Riemannian metric $\tilde{g} = g_T + \alpha\otimes \alpha$ on $H$.  Then $\kappa$ is negative, zero or positive if and only if so is the integral
\[
\int_H \kappa^{\prime}d\tilde{g} ,
\]
where $d \tilde{g}$ is the canonical Riemannian measure induced by $\tilde{g}$.
\end{proposition}

\begin{proof}
Since $H$ is totally geodesic we have  $\widetilde{\textrm{div}} n^{\prime} =0$, indeed  (see also \cite[Prop.\ 3.7, point 2]{gutierrez16})
\begin{align*}
2 \widetilde{\textrm{div}} n^{\prime}&={\textrm{Tr}(L_{n'} \tilde g)}=(L_{n'} \tilde g)(n',n')+\sum_{i=1}^n (L_{n'} \tilde g)(e_i,e_i)=2 (L_{n'} \alpha)(n')=0
\end{align*}
where $(n',e_1,\cdots, e_n)$ is an orthonormal basis for $\tilde g$, hence $\alpha (e_i)=0$ for every $i$, and where we used $L_{n'} g_T=0$.

 From equation (\ref{ren2}),
\[
\int_H \kappa^{\prime}d\tilde{g} = \int_H[e^f\kappa+n^{\prime}(f) ]d\tilde{g}.
\]
By the divergence theorem \cite[Thm.\ 16.48]{lee03},
\[
\int_Hn^{\prime}(f) d\tilde{g}=\int_H [\widetilde{\textrm{div}}(fn^{\prime}) - f\widetilde{\textrm{div}}n^{\prime}]d\tilde{g}= 0.
\]
Hence $\int_H \kappa^{\prime}d\tilde{g} =\kappa \int_H e^f d\tilde{g}$, which, since $\int_H e^f d\tilde{g}> 0$, concludes the proof.
\end{proof}

%


We are interested in the degenerate case for which the necessary condition is satisfied as the next result shows.\\

\begin{theorem}
\label{Birk app}
Suppose that the spacetime satisfies the dominant energy condition and let $n$ be any lightlike field tangent to a compact degenerate horizon $H$.
Let $\alpha:H\to T^*H$ be a 1-form  such that $\alpha(n)=1$ and let us consider the Riemannian metric $\tilde{g} = g_T + \alpha\otimes \alpha$ on $H$.
Then
\[
\int_H \kappa\,d\tilde{g} = 0.
\]
\end{theorem}

\begin{proof}
The flow $(\phi_t)_{t\in \mathbb{R}}$ of $n$ is a measure preserving flow since it preserves the Riamannian volume element, in fact, as shown in the proof of Prop.\ \ref{selc} $\widetilde{\textrm{div}}(n)=0$. Then, Birkhoff's Ergodic Theorem for measure preserving transformations gives \cite{birkhoff31,nemytskii60,arnold68}
\[
\int_H \kappa\,d\tilde{g} = \int_H \overline{\kappa}\,d\tilde{g},
\]
where $\overline{\kappa}$ is the function defined almost everywhere by
\[
\overline{\kappa}(x) = \lim_{t\rightarrow +\infty} \frac{1}{t}                   \int_0^{t}\kappa(\phi_s(x))ds.
\]
(We are not assuming the ergodicity property which would imply that $\overline{\kappa}$ is a constant.) Let $x$ be such that $\overline{\kappa}(x)$ is well defined, then $\overline{\kappa}(x) = 0$, otherwise  the function $t\mapsto \int_0^{t}\kappa(\phi_s(p))ds$ will go to $\pm \infty$, a contradiction since the horizon is degenerate (see the equivalence proved in \cite[Def.\ 4]{minguzzi21}). It follows that  $\overline{\kappa}(x) = 0$ for almost every $x$ in $H$, which implies that $ \int_H \kappa\,d\tilde{g} = 0.$
\end{proof}

By Carri\`ere's result every generator densely fills a torus, see Thm.
 \ref{nqpx}. The necessary condition restricted to the torus is satisfied  as the next result shows.\\

\begin{proposition}
\label{birk}
Let $H$ be as in Theorem \ref{Birk app}, if a generator densely fills a torus $\mathbb{T}^a$,  then
\[
\int_{\mathbb{T}^a} \kappa\,d\bar{g} = 0,
\]
where $d \bar g$ is the volume of the metric induced on the torus by $\tilde g$.
\end{proposition}

\begin{proof}
We keep the notations of the proof of Theorem \ref{Birk app}. We still denote by $n$ the restricted vector field on $\mathbb{T}^a$. It is sufficient to prove that $n$ is divergence free for the induced metric on $\mathbb{T}^a$ as the rest of the argument is as above. Let $(n,e_1, \ldots, e_{a-1})$ be an orhonormal basis of $T_p\mathbb{T}^a$  (hence $\alpha(e_i)=0$ for every $i$) and let $\overline{g}$ and  $\overline{\nabla}$ be respectively the metric and the connection induced on $\mathbb{T}^a$  by the Riemannian metric $\widetilde{g}$. Then $$\overline{div}(n) = \overline{g}(\overline{\nabla}_nn,n) + \sum_{i=1}^{a-1} \overline{g}(\overline{\nabla}_{e_i} n,e_i).$$
Since $n$ is unitary\footnote{More in detail $2 \overline{g}(\overline{\nabla}_nn,n)=n(\overline{g}(n,n))=n(\widetilde{g}(n,n))=n(1)=0$.} for $\widetilde{g}$, $\overline{g}(\overline{\nabla}_nn,n) = 0$. On the other hand, since both $n$  and $e_i$ are tangent to $\mathbb{T}^a$,
\[
\overline{g}(\overline{\nabla}_{e_i}n,e_i) = \widetilde{g}(\widetilde{\nabla}_{e_i}n,e_i) = \frac{1}{2}(L_n\widetilde{g})(e_i,e_i)=0,
\]
where in the last equality\footnote{In the first equality we used $\widetilde{\nabla}_X Y=\overline{\nabla}_XY + \beta(X,Y) u$ for $X,Y\in T\mathbb{T}^a$, where $\beta$ is the second fundamental form of $\mathbb{T}^a$ in $(H,\tilde g)$, $\overline{\nabla}_XY \in T\mathbb{T}^a$  and $u$ is a unit vector $\widetilde g$-normal to $\mathbb{T}^a$. The connection $\overline{\nabla}$ is the Levi-Civita connection for $\overline{g}$, see e.g.\ Kobayashi and Nomizu, Foundations of Differential Geometry, vol. II, Sec.\ 3.1.} we used $L_n g_T=0$ and $\alpha(e_i)=0$.
Hence $\overline{div}(n) =0$ and the restricted flow of $n$ on $\mathbb{T}^a$ is also volume preserving.
\end{proof}

\subsection{Sufficient conditions for zero surface gravity} \label{ckqp}
In this section we obtain some positive results on the existence of a field $n$ with zero surface gravity for degenerate horizons.

\subsubsection{The circle bundle case}
 In this section, we consider a horizon whose generators are all closed forming a ${S}^1$ (circle) bundle.
 Examples of this type appeared in \cite{moncrief81b,moncrief82,moncrief83}.
The main result of this section is the following.\\

\begin{theorem}
\label{closed orbits}
Let $H$ be a degenerate horizon whose generators provide the oriented fibers of a smooth ${S}^1$ (circle) bundle  $\pi: H \rightarrow B$. Then there exists on $H$ a  smooth tangent vector field $n$ of constant zero surface gravity.
\end{theorem}

\begin{proof}
 By \cite[Prop.\ 6.15]{morita01} every oriented $S^1$ bundle admits the structure of principal $S^1$ bundle, thus we can assume that $H$ is a principal $S^1$ bundle\footnote{That is, the field $n'$ can be chosen so that it induces an $S^1$-action, i.e.\ over each closed orbit the integral parameter $s$ of $n'$, $n'=\dd/\dd s$, changes of $2\pi$. In the remainder of the proof the vector field $n'$ is assumed to have been chosen in this way.}.
The idea of the proof is the following: given a smooth lightlike tangent vector field $n^{\prime}$  on $H$ inducing a $S^1$ action with period $2\pi$, we  construct a function $f$ such that $\kappa^{\prime} =  n^{\prime}(f)$. Then the existence of $n$ will follow from Eq.\ (\ref{ren2}) by choosing  $n = e^{-f}n^{\prime}$.

First, we claim that for any integral curve $\gamma$ of $n^{\prime}$,
\begin{equation}
\label{degen cond}
\int_{\gamma([0, 2\pi])}\omega^{\prime} = \int_0^{2\pi}\kappa^{\prime}(\gamma(s))ds = 0.
\end{equation}

In fact, if $ \int_{\gamma([0, 2\pi])}\omega^{\prime} \neq 0$, then  $ \int_{\gamma([0, +\infty))}\omega^{\prime} = \infty$; which means that the horizon is non-degenerate (by the equivalences of \cite[Def.\ 4]{minguzzi21}), a contradiction.

Since $\pi$ is a principal  bundle, there exists a covering $(U_i)_{i\in I}$ of $B$ by coordinate neighborhoods such that for each $i \in I$, there exists a diffeomorphism $\phi_i: \pi^{-1}(U_i) \rightarrow U_i\times {S}^1$ such that $\pi\lvert _{\pi^{-1}(U_i)} = pr_1\circ\phi_i$, where
$pr_1: U_i\times {S}^1 \rightarrow U_i$ is the first projection. Additionally, denoting with $t$ the standard coordinate on $S^1$ with period $2 \pi$,  $\phi_{i*}n^{\prime} = \frac{\partial}{\partial t}$.
Let  $x^1,\cdots, x^{n-1}$ be coordinates  on $U_i$.
Let us define $\tilde{f}_i$ on  $U_i\times {S}^1$ by $\tilde{f}_i(x^1,\cdots, x^{n-1},t) = \int_0^t \kappa^{\prime}(\phi_i^{-1}(x^1,\cdots, x^{n-1},s))ds.$ We note that $\tilde{f}_i$ is well defined, i.e, it is periodic with respect to the variable $t$. This follows from Eq.\ (\ref{degen cond}).

Moreover, by definition, $(\phi_{i*}n^{\prime})(\tilde{f}_i) =  \kappa^{\prime}\circ\phi_i^{-1}$, hence on $\pi^{-1}(U_i)$, it holds $n^{\prime}(f_i) = \kappa^{\prime}$ where $f_i = \tilde{f}_i\circ \phi_i.$ Let $(\overline{\rho}_i)_{i\in I}$   a partition of unity of $B$ subordinate to $(U_i)_{i\in I}$ and consider  the partition of unity  $(\rho_i)_{i\in I}$ of $H$ subordinate to $(\pi^{-1}(U_i))_{i\in I}$ where $\rho_i =  \overline{\rho}_i\circ\pi$; then $ f = \sum_{i\in I} \rho_if_i$ satisfies  on $H$, $n^{\prime}(f) = \kappa^{\prime}$. It follows that $n = e^{-f}n^{\prime}$ has constant zero surface gravity.
\end{proof}

\subsubsection{Horizons without closed orbits}
In Theorem \ref{closed orbits}, we proved that for compact degenerate horizons that are principal circle bundles (a subcase of  case (i) of Theorem \ref{car}) there exists a smooth lightlike vector field $n$ with zero surface gravity. In the next theorem we prove that this is also the case if (iii-a) or (iv) holds under some algebraic condition on the flow.

The real vector space $\mathbb{R}^n$ will be equipped with
its usual scalar product $\langle\, , \rangle$  and the associated norm $\lvert . \lvert$.\\

\begin{definition}(\cite[Def.\ 2.2]{dehghan07})
Let $a = (a_1, \cdots , a_n)$ be a vector of $\mathbb{R}^n$ such that the numbers $a_1, \cdots , a_n$  are linearly independent over $\mathbb{Q}$. (This implies in particular that the orbits of the linear flow defined by the  projection of $a$ on $\mathbb{T}^n = \mathbb{R}^n / \mathbb{Z}^n $ are dense in $\mathbb{T}^n$).
We say that the vector  $a$ is {\em Diophantine} if there exist real numbers $C > 0$ and $\tau > 0$ such that $\lvert \langle m, a \rangle  \lvert \geq \frac{C}{\lvert m \lvert^{\tau}}$ for any non zero $m \in \mathbb{Z}^n$.
\end{definition}

For instance, any vector $ a = (a_1, \cdots , a_n)$  of $\mathbb{R}^n$ (such that the numbers $a_1, \cdots , a_n$  are linearly independent over $\mathbb{Q}$) for which the components $a_1, \cdots , a_n$ are algebraic numbers is Diophantine.\\

\begin{definition}
We will say that a linear flow on  $\mathbb{T}^n$ is {\em Diophantine} if the vector $a$ of $\mathbb{R}^n$ defining it is a Diophantine vector.\\
\end{definition}

We have the following\footnote{The statement and proof of the next result are improved and simplified with respect to the published version. }
\begin{theorem} \label{cqjb}
Let $(M,g)$ be a  $4$-dimensional spacetime satisfying the dominant energy condition. Let $H$ be  a smooth degenerate compact totally geodesic horizon in $M$.  Suppose the flow on the horizon is of the type (iii-a) or (iv) of Theorem \ref{car} hence conjugate to a linear flow on a 3-torus. Moreover, assume the conjugate flow on $\mathbb{T}^3$  is Diophantine  (in case (iv)) or Diophantine  on one (and hence every) $2$-torus fibers (in case (iii-a)). Then
there exists a smooth lightlike tangent vector field $n^{\prime}$ on $H$ with constant zero surface gravity.
\end{theorem}

\begin{proof}
Since the flow is conjugate (without parameter) to a linear flow on $\mathbb{T}^3$ we can work on this manifold and rescale $n$ on $H$ so that its push forward on  $\mathbb{T}^3$  is the vector defining the linear flow.  Similarly the surface gravity $k$ can be pushed to a function on  $\mathbb{T}^3$ and the bilinear form $g_T$ can be pushed to a bilinear form on $\mathbb{T}^3$. We are thus going to work on $\mathbb{T}^3$, dropping the diffeomorphism, where $n$ is a constant vector (where $\mathbb{T}^3$ is given the canonical flat affine structure induced from $\mathbb{R}^3$). Note that $g_T$ is invariant by the flow of $n$ and so, introduced a constant 1-form $\alpha$ such that $\alpha(n)=1$, we have $L_n\alpha=0$ and defining $\widetilde g=g_T+\alpha\otimes \alpha$, $L_n \widetilde g=0$ and hence $L_n \dd \widetilde g=0$.

In case (iv), as $n$ is Diophantine, its components are rationally independent (meaning that a linear combination with integer coefficients vanishes only if the coefficients are all zero), which implies unique ergodicity\footnote{In the map case see, for instance,  M. Einsiedler and T. Ward, {\em Ergodic Theory with a view towards Number Theory} (Springer-Verlag, London, 2011), Cor.\ 4.15. For the flow case that interests us see, e.g.\, J.-F.\ Quint, {\em Examples of unique ergodicity of algebraic
flows}, Lectures at Tsinghua University, Beijing, November 2007, Example 1.2.11.
 The flow case can be proved imposing for any continuous function $f: \mathbb{T}^n \to \mathbb{R}$ and any $t \in \mathbb{R}$,
$\int_{\mathbb{T}^n} f(\phi_t(x)) \, d\mu(x) = \int_{\mathbb{T}^n} f(x)  d\mu(x)$, and proceeding
with a Fourier expansion of $f$. For related material see V. I. Arnold {\em Mathematical Methods of Classical Mechanics}    (Springer-Verlag, New York, 1989), Sec.\ 51.}
and hence that $\dd \widetilde g =a \dd x \dd y \dd z$, where $a$ is a constant and $\dd x \dd y \dd z$ is the standard flat volume on the torus. From Theorem \ref{Birk app} we get that $\int_{H} \kappa   \dd x \dd y \dd z=0$ which, from \cite[Thm.\ 2.3 point (i)]{dehghan07} implies that there exists a smooth function $f$ such that $n(f)=\kappa$.

In case (iii-a), $n$ is Diophantine when regarded as tangent vector over the $\mathbb{T}^2$ fibers, which implies that the flow is uniquely ergodic on each $\mathbb{T}^2$ fiber. The canonical volume on such tori thus coincides, up to a factor, with that of the  metric induced from $\widetilde g$ and so, by  our Prop.\ \ref{birk},  on each 2-torus the integral of $\kappa$ with the canonical volume of the 2-torus vanishes.
By \cite[Thm.\ 2.4]{dehghan07} there exists a smooth function $f$ such that $n(f)=\kappa$.

In both cases it follows that   $n^{\prime} = e^{-f}n$ has constant zero surface gravity.
\end{proof}%

\section{Further results for non-degenerate horizons: vanishing of the 2-form $d\omega$}
In the next theorem, we prove that the  2-form $d\omega$ vanishes identically on non-degenerate compact horizon $H$ with dense null generators (and hence $H$ is a torus)  in $4$-dimensional spacetime. For horizons, whose  every generator densely fills a two-torus (and so are torus bundle over the circle \cite[Cor.\ 1.2]{bustamante21}), we prove that the 2-form $d\omega$ vanishes at least on a two-torus. As the proof uses some results about contact $1$-form, we recall the following (see \cite{geiges08}).\\

\begin{definition}
Let $M$ be a manifold of odd dimension $2n+1$. An orientable {\em contact
structure} is a maximally non-integrable hyperplane field $\xi = ker\alpha \subset TM$, that is, the defining differential $1$-form $\alpha$ is required to satisfy $\alpha\wedge(d\alpha)^n \neq 0$ (meaning that it vanishes nowhere). Such a $1$-form $\alpha$ is called a contact form. The pair $(M,\xi)$ is called a {\em contact manifold}.\\
\end{definition}

It follows that $\alpha\wedge(d\alpha)^n$ is a volume form on $M$ and so $M$ is orientable.
In dimension $3$, the contact condition reduces to  $\alpha\wedge(d\alpha) \neq 0$ which is equivalent to $d\alpha\lvert_{\xi} \neq 0.$

Associated with a contact form $\alpha$, one has the
so-called {\em Reeb vector field} $R_{\alpha}$, uniquely defined by the equations
\begin{itemize}
\item[i)] $d\alpha(R_{\alpha}, .) = 0$,
\item[ii)] $\alpha(R_{\alpha}) = 1$.
\end{itemize}

The Reeb vector field of a contact form defined on a $3$-dimensional compact manifold has always a periodic orbit \cite{taubes07}. This fact allows us to prove the following.\\

\begin{proposition}
Let $(M, g) $ be a $4$-dimensional spacetime which satisfies the dominant energy condition.
Let $H$ be a compact non-degenerate horizon of  $(M, g)$.
\begin{enumerate}
\item If the null generators are dense in $H$,  then the  2-form $d\omega$ vanishes identically.
\item If every generator densely fills a two-torus, then the 2-form $d\omega$ vanishes at least on a two-torus.
\end{enumerate}
\end{proposition}

\begin{proof}
From \cite[Thm.\ 7]{minguzzi21}, we can choose a smooth  normal tangent vector field $n$ such that $\nabla_nn = -n$. From \cite[Lemma 5, Lemma 6]{minguzzi21}, $d\omega(n, .) = 0$, $L_nd\omega = 0$. So $d\omega$ is invariant under the flow of $n$.  Suppose that, $\forall p\in H$, $d\omega_p \neq 0$. Since $\omega(-n) = 1$, $d\omega(-n, \cdot ) = 0$, it follows that $\omega$ is a contact form on the compact $3$-dimensional manifold $H$ with Reeb vector field $-n$. From \cite{taubes07}, $-n$ and then $n$ has at least a periodic orbit. So, if  the null generators are dense or densely fill two-tori, then there exists some $p\in H$ such that $d\omega_p = 0$. By the invariance of $d\omega$ under the flow of $n$, it follows that $d\omega$ vanishes along the orbit $\gamma_p$ of $n$ and by continuity, it vanishes on its closure. If the null generators are dense, then the  2-form $d\omega$ vanishes identically on $H$, while if every generator densely fills a two-torus, then the 2-form $d\omega$ vanishes at least on a two-torus.
\end{proof}

\begin{corollary}
Let $(M, g) $ be a $4$-dimensional spacetime which satisfies the dominant energy condition.
Let $H$ be a compact non-degenerate horizon of  $(M, g)$.
 If the null generators are dense in $H$ (hence it is $\mathbb{T}^3$) then
 \begin{enumerate}
\item The $1$-form $\omega$ defines a foliation on $H$.
\item $\forall X, Y \in TH$,
 $R(X, Y )n = 0$ and    $R(n,X)Y = R(n,Y)X$.
 \end{enumerate}
\end{corollary}

\begin{proof}
Point 1 is trivial since $\omega$ is a non-singular closed $1$-form.
Point 2 follows directly  from  Lemma 3 of \cite{minguzzi21}.

\end{proof}

\section{Conclusions}

In this work, we pointed out that the structure of compact (Cauchy) horizons in general relativity can be studied by taking advantage of a body of results developed since the fifties of the last century for Riemannian foliations.

Some of these results were rediscovered over the years by the community working in Lorentzian geometry, but several others remained unnoticed. We introduced the reader to this body of literature selecting some first relevant results for their application to horizons.

It should be noted that in relativity theory the study of degenerate horizons remained elusive because they lead to flows that can be, contrary to the non-degenerate case, non-isometric. In the non-degenerate case, one can apply the full machinery of Lie group theory via the Myers-Steenrod theorem, which is not available in the non-isometric case. Still, as we observed, the more general theory of Riemannian foliations can indeed be applied to the degenerate case. The results on Riemannian flows, suitably reinterpreted, are independent of any non-degeneracy condition. In this way, by using results by Molino and Carri\`ere, we  reobtained  the classification by Bustamante and Reiris for the non-degenerate 4-dimensional spacetime case and showed that the more general (possibly degenerate) case just involves one further possible hyperbolic structure $\mathbb{T}^3_A$.

For a compact Cauchy horizon of a simply connected partial Cauchy hypersurface in arbitrary dimensions  the flow turns out to be isometric (thanks to results by Ghys).
The study of Riemannian flows is  so advanced that  detailed structure results for horizons can be given for spacetime dimension 5 (thanks to results  by Almeida and Molino).

In this paper we also discussed the existence of Killing fields tangent to the horizon but transverse to the generators, again importing and adapting results from the theory of Riemannian flows.

Although the notion of surface gravity does not belong to the geometry of Riemannian flows,  the found structure for the horizon allows one to approach the remaining problem of the existence of a field of  zero surface gravity in the degenerate case via a case-by-case study. In   Section \ref{ckqp} we have indeed obtained some first results in this direction  showning the importance of the Diophantine condition.

Finally, in the last section, we returned to the non-degenerate case. Since much of the geometry of the horizon is controlled by the 1-form $\omega$ we provided some novel results on this form and on the geometry of the horizon that might be useful in future work.

\section*{Acknowledgments}
R.H. was supported by ICTP-INDAM ``Research in Pairs'' grant and by ``fondi d'internazionalizzazione''  of the Department of Mathematics of Universit\`a Degli Studi di Firenze thanks to which he visited Florence in August-November 2021, during the covid pandemic.

This study was funded by the European Union - NextGenerationEU, in the framework of the PRIN Project (title) {\em Contemporary perspectives on geometry and gravity} (code 2022JJ8KER – CUP B53D23009340006). The views and opinions expressed in this article are solely those of the authors and do not necessarily reflect those of the European Union, nor can the European Union be held responsible for them.

We thank Gianluigi Del Magno for help in finding  some references on unique ergodicity of flows on tori.




\end{document}